\documentclass[aps,pra,twocolumn,showpacs,preprintnumbers,amsmath,amssymb,footinbib]{revtex4}
\usepackage{graphicx,epsfig}
\usepackage{bm}
\usepackage{dcolumn}
\usepackage{xcolor}
\usepackage[breaklinks=true,colorlinks,citecolor=blue,linkcolor=blue,urlcolor=blue]{hyperref}
\usepackage{lineno,hyperref}
\usepackage{amsmath,amssymb,amsfonts}
\usepackage{algorithmic}
\usepackage{textcomp}
\usepackage{float}
\usepackage{epstopdf}
\usepackage{subfigure}
\newtheorem{theorem}{Theorem}{}
\newtheorem{assumption}{Assumption}{}
\newtheorem{lemma}{Lemma}{}
\newtheorem{proposition}{\textbf{Proposition}}
\newtheorem{remark}{Remark}{}
\newenvironment{proof}{{\noindent\it Proof.}\quad}{\hfill $\square$\par}

\begin{document}
\title{Distributed Leader-Follower Formation Tracking Control of Multiple Quad-rotors}
%
\author{Lixia Yan}
\thanks{yanlixia@buaa.edu.cn;mabaoli@buaa.edu.cn}
\author{Baoli Ma}

\address{The Seventh Research Division, School of Automation Science and Electrical Engineering, Beihang University, 100191, Beijing, China}

\date{\today}

\begin{abstract}
The leader-follower formation control analysis for multiple quad-rotor systems is investigated in this paper. To achieve predefined formation in the three-dimensional air space ($x,y$ and $z$), a novel local tracking control law and a distributed observer are obtained. The local tracking control law starts with finding a bounded continuous yet greater-than-zero control in $z$, based on which following a feedback linearization (FL) controls derived for errors associated with $x$ and $y$. By this design method, we obtain less states to be regulated than the traditional extension FL methodology. Then, the proposed distributed observer solves the problems that only a subset of followers can know the leader’s states and only neighboring communication is available. Simulation results validate the proposed formation scheme.
\end{abstract}
\maketitle


\section{Introduction}
A quad-rotor is a multi-rotor helicopter that is lifted and propelled by four symmetrically mounted rotors \cite{RN413,RN328,RN335,RN412}. It has experienced a great boom in recent years due to potential applications such as aerial photography, geological exploration and disaster relief. To enhance the reliability and safety of these applications, researchers have developed many cooperative strategies of quad-rotor systems, one of which is the leader-follower formation scheme that allows for steering multiple quad-rotors to form a geometric pattern while tracking a leader/reference. For this sake, three classical approaches can be applied, that is, linearization  \cite{RN849,RN877,RN851,RN857,RN876,RN885}, inner-outer loop method \cite{RN863,RN858,RN874,RN855,RN848,RN843,RN856} and feedback linearization approach \cite{RN435,RN417,RN412,RN888}.\\
\indent A direct approach for quad-rotor formation consists in linearizing quad-rotor model around maneuvering point. Due to easy implementation of linearized quad-rotor model, the control scheme reported in \cite{RN849} elaborates the potential functions with its formation control law, achieving the formation pattern with collision avoidance behavior. The control design reported in  \cite{RN877} realizes cooperative formation of multiple heterogeneous agents, including many quad-rotors and linearized differentiable mobile robots moving on the ground. Given the possible interaction fault between adjacent quad-rotors during their formation tracking, two $H_\infty$-formation schemes are proposed in \cite{RN851} and \cite{RN876} respectively, presenting fault-tolerant capacity during flights of multiple quad-rotors. In \cite{RN857}, the result developed for high-order linear integrators is adapted to solve the formation problem of multiple quad-rotors with external disturbances. A finite-time formation tracking controller can be found in \cite{RN885}. Some other literature, such as \cite{RN845,RN875,RN466,RN847}, extend the results developed for linear double integrators to achieve formation tracking of multiple quad-rotors directly. These formation tracking control laws, developed by either linearized quad-rotor model or linear double integrators, however, can only solve the formation rendezvous problem or formation tracking problem with a slowly-moving leader. They are incapable of steering quad-rotors to perform agile motions with large roll/pitch angles due to the loss of model nonlinearities in their designs.\\
\indent As for the inner-outer loop approach, the longitudinal and latitudinal position errors therein are viewed as an outer loop, and the attitude errors in roll and pitch are called inner loop \cite{RN852}. This approach becomes popular out of two facts. First, the altitude and yaw can be steered independently. Second, the roll and pitch angles can be viewed as virtual control inputs for the dynamics of longitudinal and latitudinal position. Some associated results can be found in \cite{RN863,RN858,RN874,RN855,RN848,RN843,RN856}. In \cite{RN863}, a discontinuous observer for the formation trajectory is proposed based on neighboring connections and graph theory, which, together with the inner attitude algorithm, makes the formation errors converge to zero asymptotically. A discontinuous formation tracking controller reported in \cite{RN874} obtains finite-time convergence of the quad-rotor formation errors. The centroid formation, steering the average position of all quad-rotors to track the leader's position, can be achieved by the control laws proposed in both \cite{RN863} and \cite{RN874}. In consideration of inefficiency of GPS during indoor flying, the works \cite{RN858,RN855} propose two vision-based formation control laws. The invertibility of the Laplacian matrix associated with an undirected connected interaction graph is made full use by \cite{RN848}, in which the quad-rotor formation error is proven to be convergent, and this convergent rate is proportional to the smallest eigenvalue of the interaction graph. The non-smooth consensus formation tracking scheme shown in \cite{RN843} achieves the formation pattern with a constant speed. For agile coordination, a virtual structure approach utilized in \cite{RN856} views each quad-rotor in the group as a rigid body, achieving swarm with over 200 quad-rotors. However, the inner-outer loop formation schemes of multiple quad-rotors generally lead to difficulties in obtaining the desired attitude derivatives for the inner-loop control. This is because the desired attitude (roll and pitch angles) includes the coordination position and velocities, and the direct calculation of desired attitude derivatives will contain unavailable information of unconnected quad-rotor. To obtain the desired attitude derivatives, the velocity observer approach \cite{RN863}, first derivative method \cite{RN874} or direct differentiation \cite{RN858} can be applied. As a result, the overall stability analysis becomes incomplete.\\
\indent The feedback linearization approach involves coming up with system state transformations into an equivalent linear system through variable changes, state extensions and suitable control input. Although it is proved that the normal twelve-dimensional system of a quad-rotor is not feedback linearizable in \cite{RN435}, an extended system with fourteen system states is feedback linearizable with triple fourth-order position states and a second-order yaw state. This kind of dynamic extension is built upon the flat property of the quad-rotor and viewing the second-order state of thrust force as the control input to be designed. Classical results associated with this approach can be found in \cite{RN417,RN412}. The main inefficiency of feedback linearization based on dynamic extension lies in that the greater-than-zero total thrust control cannot be always ensured, which might obstruct the flight because the quad-rotor needs an upward thrust to hover in the air. To solve the problem for the dynamic-extension-based feedback linearization method associated with quad-rotor formation, a recent literature \cite{RN888} proposes a control design based on Euler-Poincar$\mathrm{\acute{e}}$ equations, by whose result the obtained thrust force control is kept being greater than zero all the time.\\
\indent Motivated by the facts and challenges stated above, this paper makes further endeavors to consider the leader-follower formation issue for a team of quad-rotors. To deal with this formation problem, we first establish a local tracking control law via non-regular feedback linearization method, given any reference signal with bounded derivatives and a reference altitude acceleration no more than gravitational acceleration. Then, a distributed observer is investigated by the reduced-order and linear time-varying techniques, solving the problem that only neighboring connection is available and only partial followers can know the leader’s states. The combination of the local tracking control law and the distributed reduced-order observer leads to the formation tracking scheme. Compared with previous research, the main innovation points of proposed formation scheme are as follows:
\begin{itemize}
  \item the safe maneuvering can be ensured as the total thrust of each quad-rotor is kept being greater than zero all the time, and simultaneously, the roll and pitch are strictly limited in $(-\pi/2,\pi/2)$;
  \item the local tracking controller derived from non-regular feedback linearization method allows quad-rotors for admissible agile motions with large roll and pitch angles;
  \item the fully distributed coordination without global interaction will not cause large communication burden when adding cooperative quad-rotors to achieve complex formation pattern.
\end{itemize}
\indent The rest is organized as follows. Section 2 contains quad-rotor modeling, basic graph theory and problem formulation. Section 3 presents the main results. Section 4 considers the numerical simulation. Section 5 concludes the work briefly.\\
\textbf{Notations:}~The norm '$\|\cdot\|$ refers to Euclidean norm, the letter '$e$' without subscript/superscript denotes exponent, '${\mathrm{diag}}(\cdot)$' means diagonalization and $I_n$ denotes a $n$-dimensional identity matrix.
\section{Preliminaries and Problem Formulation}
\subsection{Model Description}
\begin{figure}[H]
\centering     
\includegraphics[width=.65\columnwidth]{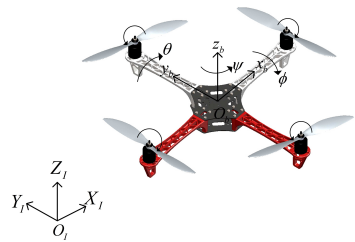}
\caption{The illustration of a quad-rotor.}
\end{figure}
The quad-rotor model used for control formulation and validation displays a symmetrical configuration, see an example in Figure 1, where $F^I(O_Ix_Iy_Iz_I)$ denotes the inertial frame and $F^B(x_by_bz_b)$ denotes the body-fixed frame. Suppose that there are $n$ quad-rotors with index belonging to $\mathcal{N}=\{1,2,...,n\}$. For $i\in \mathcal{N}$, let $p_i=[x_i,y_i,z_i]^T$ be the position of the $i$-th quad-rotor in inertial frame, and $\eta_i=[\phi_i,\theta_i,\psi_i]^T$ denote the roll angle $\phi_i$, pitch angle $\theta_i$ and yaw angle $\psi_i$, respectively. The common assumption below can be used to simplify the quad-rotor modeling.
\begin{assumption}\label{A1ds}
The quad-rotor body is rigid and with an invariant structure and aerodynamic parameters, and the thrust and drag are proportional to the square of the propeller speed.
\end{assumption}
Based on Assumption $\ref{A1ds}$, the quad-rotor dynamics can be described by \cite{RN413,RN885}
\begin{equation}\label{DynaQd}
\left\{ \begin{split}
&\ddot p_i = R_{i}^I{T_{iB}}/m_i + {G}\\
&M_i\left( \eta_i  \right)\ddot \eta_i  + C_i\left( {\eta_i ,\dot{\eta}_i } \right)\dot{\eta}_i = {\tau_{ic}}
\end{split} \right.,
\end{equation}
where $T_{iB}=[0,0,F_i]^T$ is the thrust with respect to the body-fixed frame with $F_i$ the total lift, $m_i$ is the mass, $G=[0,0,-g]^T$ denotes the gravity vector with respect to the inertial frame and $\tau_{ic}=[\tau_{i,\phi},\tau_{i,\theta},\tau_{i,\psi}]^T$ is the control torque. Moreover, the specific definitions of $M_i(\eta_i)$ and $C_i(\eta_i,\dot{\eta}_i)$ can be found in \cite{RN413,RN885}.\\
\indent To present the main idea concisely, define virtual control inputs by $u_{i,1}=F_i/m_i$ and $\tau_{i}=[u_{i,2},u_{i,3},u_{i,4}]^T$, where $\tau_i=M^{-1}_i(\eta_i)(\tau_{ic}-C_i(\eta_i,\dot{\eta}_i)\dot{\eta}_i)$. Then, rewrite the quad-rotor dynamics as follows:
\begin{equation}\label{SimpModel}
\left\{ \begin{split}
{{\ddot x}_i} &= {u_{i,1}}\left( {\cos {\psi _i}\sin{\theta _i}\cos{\phi _i} + \sin{\psi _i}\sin{\phi _i}} \right)\\
{{\ddot y}_i} &= {u_{i,1}}\left( {\sin {\psi _i}\sin {\theta _i}\cos {\phi _i} - \cos {\psi _i}\sin {\phi _i}} \right)\\
{{\ddot z}_i} &= {u_{i,1}}\cos {\theta _i}\cos {\phi _i} - g\\
{{\ddot \phi }_i} &= {u_{i,2}}\\
{{\ddot \theta }_i} &= {u_{i,3}}\\
{{\ddot \psi }_i} &= {u_{i,4}}
\end{split} \right..
\end{equation}
Suppose that the leader agent with index $0$ is time-parameterized and defined by
\begin{equation}\label{reftrj}
p_0(t)=[x_0(t),y_0(t),z_{0}(t)]^T,
\end{equation}
where $(x_0,y_0)$ denotes the coordinate of latitude and longitude and $z_0$ is the altitude. Steering multiple quad-rotors to form a pattern while tracking a leader is related to many potential applications such as cooperative patrolling and geometrical prospecting, it is therefore very important to maintain quad-rotors moving with fixed altitude in regards of safety and airspace limitations, which motivates the assumption below.
\begin{assumption}\label{sssref}
The latitudinal and longitudinal positions of the leader, $x_0$ and $y_0$, are fourth-order differentiable with bounded derivatives; and the altitude $z_0$ is a constant.
\end{assumption}
\subsection{Graph Theory}
A graph $\mathcal{G}=\{\mathcal{N},\mathcal{E},\mathcal{A}\}$ is used to describe the interaction among multiple quad-rotors, where $\mathcal{N}=\{1,2,...,n\}$ denotes the node set, ${\mathcal{E}} \subseteq \mathcal{N} \times \mathcal{N}$ is the edge set and ${\mathcal{A}}$ is adjacent matrix \cite{RN272}. Each node $i\in \mathcal{N}$ represents one quad-rotor, and an edge $\left\{ {\left( {i,j} \right):i \ne j} \right\} \in \mathcal{E}$ denotes that the quad-rotor $j$ can send information to quad-rotor $i$ via wireless module. The adjacent matrix is defined by $\mathcal{A }= \left\{ {{a_{ij}}} \right\} \in {\mathbb{R}^{n \times n}}$, where ${a_{ij}} = 1$ if $\left( {i,j} \right) \in{ {\mathcal{E}}}$, otherwise ${a_{ij}} = 0$. Self connection is forbidden by setting ${a_{ii}} = 0,\forall i \in \mathcal{N}$. For an undirected graph, ${a_{ij}} = 1 \Leftrightarrow {a_{ji}} = 1$ holds, denoting that the quad-rotor $i$ and quad-rotor $j$ can transmit information to each other. A path of graph ${\mathcal{G}}$ is an edge sequence $\{(i,j_1),(j_2,j_3),...,(j_*,j)\}$. The in-degree matrix of graph ${\mathcal{G}}$ is given by
$
\mathcal{{D}}{\rm{ = \mathrm{diag}}}\left\{ [l_{11},l_{22},...,l_{nn}] \right\},{l_{ii}} = \sum\limits_{j = 1}^n {{a_{ij}},\forall i,j \in \mathcal{N}},
$
and the Laplacian matrix can then be obtained as
\begin{equation}\label{lapm}
{\mathcal{L}}={\mathcal{D}}-{\mathcal{A}}.
\end{equation}
As reported in \cite{RN531}, the matrix $\mathcal{L}$ is semi-positive definite and has only one zero eigenvalue and $n-1$ positive eigenvalues provided that $\mathcal{G}$ is undirected and connected. Define $a_{i0}=1$ if there is a valid information flow from the leader to the $i-$th quad-rotor, otherwise $a_{i,0}=0$, which then leads to the matrix given by
\begin{equation}\label{hHH}
\begin{split}
\mathcal{H}&=\mathcal{L}+\mathcal{B},
\end{split}
\end{equation}
where $\mathcal{B}=\text{diag}\{[a_{10},...,a_{n0}]\}$. It has been shown in \cite{RN531} that $\mathcal{H}$ is positive definite if $\mathcal{G}$ is connected and the matrix $\mathcal{B}$ is non-trivial. Moreover, $\mathcal{H}$ is symmetric if $\mathcal{G}$ is undirected. For a basic formation setup, the assumption below is needed.
\begin{assumption}
The graph $\mathcal{G}$ is undirected and connected, and $\mathcal{B} \neq \mathrm{0}_{n\times n}$.
\end{assumption}
\begin{remark}
$\mathcal{B}\neq 0$ means that there at least one quad-rotor can know the leader's position and derivatives up to appropriate orders.
\end{remark}
\subsection{Problem Formulation}
In this note, the focus is set on achieving a fixed formation pattern. Define a constant vector by
\begin{equation}\label{relp}
\Delta_i=[d_{i,x},d_{i,y},d_{i,z}]^T,
\end{equation}
and formation error by
\begin{equation}\label{zetaeee}
{\eta _i} = \left[ \begin{array}{l}
{x_i} - {x_0} - {d _{i,x}}\\
{y_i} - {y_0} - {d _{i,y}}\\
{z_i} - {z_0} - {d _{i,z}}
\end{array} \right].
\end{equation}
The control objective can then be stated as:\textit{ Based on the quad-rotor model $(\ref{SimpModel})$ and Assumptions 1-3, find control laws $(u_{i,1},u_{i,2},u_{i,3},u_{i,4})$ so that} $\mathop {\lim }\limits_{t \to \infty } {\eta _i} = 0, \forall i \in \mathcal{N}$.
\section{The main results}
The formation scheme includes a local controller and a distributed observer. Given any smooth reference trajectory with bounded derivatives and a less-than-$g$ altitude acceleration, the local controller is proposed firstly with the help of non-regular feedback linearization technique, steering the tracking errors converge to zero asymptotically. The distributed observer is then investigated via interaction between connected agents and is viewed as virtual reference trajectory. The leader-follower formation can be realized via applying the local control law on each follower quad-rotor to track its reference signal, see Figure 2 for an illustration.
\begin{figure}[H]
\centering     
\includegraphics[width=.58\columnwidth]{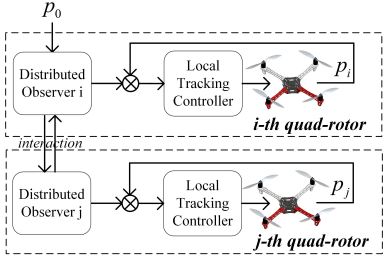}
\caption{The illustration of formation control scheme.}
\end{figure}
\subsection{Local Tracking Control Design}
A lemma is needed to formulate the control design.
\begin{lemma}\cite{RN842}
The system
\begin{equation}\label{l2l2}
\ddot{\xi}=-a_1\tanh(\dot{\xi}+a_2\xi)-a_3\tanh(\dot{\xi})
\end{equation}
is globally asymptotically stable with respect to $\xi\in \mathbb{R}^m$ and $\dot{\xi}\in\mathbb{R}^m$, provided that $a_1,a_2,a_3>0$.
\end{lemma}
Let $p_{id}=[x_{id},y_{id},z_{id}]^T$ be the reference trajectory of the $i$-th quad-rotor and suppose that $p_{id}$ is fourth-order differentiable with bounded derivatives and $|\ddot{z}_{id}(t)|<g,\forall t \geq 0$. The assumption $\ddot{z}_{id}(t)<g$ is made here due to concern that the altitude coordinations between quad-rotors will be addressed and $z_{id}$ can be time-varying. Define errors
\begin{equation}\label{e1e2e3}
\left\{
\begin{split}
e_{i,x}&=x_i-x_{id}-d_{i,x}\\
e_{i,y}&=y_i-y_{id}-d_{i,y}\\
e_{i,z}&=z_i-z_{id}-d_{i,z}
\end{split}\right.,
\end{equation}
and calculate the second-order derivative of $(\ref{e1e2e3})$ as follows,
\begin{equation}\label{de123}
\left\{ \begin{split}
{{\ddot e}_{i,x}} &= {u_{i,1}}\left( {\cos {\psi _i}\sin{\theta _i}\cos{\phi _i} + \sin{\psi _i}\sin{\phi _i}} \right) - {{\ddot x}_{id}}\\
{{\ddot e}_{i,y}} &= {u_{i,1}}\left( {\sin {\psi _i}\sin {\theta _i}\cos {\phi _i} - \cos {\psi _i}\sin {\phi _i}} \right) - {{\ddot y}_{id}}\\
{{\ddot e}_{i,z}} &= {u_{i,1}}\cos {\theta _i}\cos {\phi _i} - g - {{\ddot z}_{id}}
\end{split} \right..
\end{equation}
In view of $(\ref{de123})$, the non-regular feedback linearization method can be applied. The specific steps are:
\begin{enumerate}
  \item Determine $u_{i,1}$ and $u_{i,4}$;
  \item Basing on the designed $u_{i,1}$ and the closed-loop $e_{i,z}$-dynamics, compute and obtain the feedback linearizable form of $[e_{i,x}^{(4)},e_{i,y}^{(4)}]$, in which the control inputs $(u_{i,2},u_{i,3})$ will be obviously contained;
  \item Find a suitable feedback linearized control law for $(e_{i,x},e_{i,y})$-dynamics;
\end{enumerate}
It is admissible to steer the yaw angle $\psi_i$ to zero for formation purpose \cite{RN874,RN848} and determine $u_{i,4}$ in a PD-form by
\begin{equation}\label{yawctrl}
{u_{i,4}} =  - {k_{1,\psi }}{\psi _i} - {k_{2,\psi }}{{\dot \psi }_i},
\end{equation}
where $k_{1,\psi},k_{2,\psi}>0$. Suppose that $\theta_i,\phi_i\in(-\pi/2,\pi/2)$ and design $u_{i,1}$ by
\begin{equation}\label{ui14}
 \begin{split}
{u_{i,1}} &= \frac{1}{{\cos {\theta _i}\cos {\phi _i}}}\bar{u}_{i,1},\\
\bar{u}_{i,1}&:= g + {{\ddot z}_{id}} - {k_{1,z}}\tanh \left( {{{\dot e}_{i,z}} + {k_{2,z}}{e_{i,z}}} \right)- {k_{3,z}}\tanh {{\dot e}_{i,z}},
\end{split}
\end{equation}
where $k_{z,1},k_{z,2},k_{z,3}>0$ and $k_{z,1}+k_{z,3}<g-|\ddot{z}_{id}|$. The term $\bar{u}_{i,1}$ is an intermediate variable for calculation convenience. Via substituting $(\ref{ui14})$ into $(\ref{de123})$, $\ddot{e}_{i,z}$ becomes
\begin{equation}\label{dezez}
\begin{split}
\ddot{e}_{i,z}&=-k_{1,z}\tanh(\dot{e}_{i,z}+k_{2,z}e_{i,z})-k_{3,z}\tanh(\dot{e}_{i,z}),\\
\end{split}
\end{equation}
which shows that, according to Lemma 1, $e_{i,z}$ is asymptotically stable. Additionally, it can be seen that $\bar{u}_{i,1}>0,\forall t \geq 0$, which helps to formulate the control inputs $[u_{i,2},u_{i,3}]^T$. By $(\ref{ui14})$, organize the column $[\ddot{e}_{i,x},\ddot{e}_{i,y}]^T$ by
\begin{equation}\label{ddexy}
\left[ \begin{array}{l}
{{\ddot e}_{i,x}}\\
{{\ddot e}_{i,y}}
\end{array} \right] = {{\bar u}_{i,1}}{R_2}\left( {{\psi _i}} \right)S\left[ \begin{array}{l}
\tan {\theta _i}\\
\displaystyle\frac{{\tan {\phi _i}}}{{\cos {\theta _i}}}
\end{array} \right] - \left[ \begin{array}{l}
{{\ddot x}_{id}}\\
{{\ddot y}_{id}}
\end{array} \right],
\end{equation}
where
\begin{equation}\label{barui1}
\begin{split}
{R_2}\left( {{\psi _i}} \right) &= \left[ {\begin{array}{*{20}{c}}
{\cos {\psi _i}}&{ - \sin {\psi _i}}\\
{\sin {\psi _i}}&{\cos {\psi _i}}
\end{array}} \right],S = \left[ {\begin{array}{*{20}{c}}
1&0\\
0&{ - 1}
\end{array}} \right].
\end{split}
\end{equation}
Differentiating $(\ref{ddexy})$ results
\begin{small}
\begin{equation}\label{thirdorder}
\begin{split}
\left[ \begin{array}{l}
e_{i,x}^{\left( 3 \right)}\\
e_{i,y}^{\left( 3 \right)}
\end{array} \right] &= \underbrace {\left[{{\dot {\bar u}}_{i,1}}{R_2}\left( {{\psi _i}} \right) + {{\bar u}_{i,1}}\frac{{\mathrm{d}{R_2}\left( {{\psi _i}} \right)}}{{\mathrm{d}t}}\right]S\left[ \begin{array}{l}
\tan {\theta _i}\\
\displaystyle\frac{{\tan {\phi _i}}}{{\cos {\theta _i}}}
\end{array} \right]}_{{\Xi _{i,1}}}+ {{\bar u}_{i,1}}{R_2}\left( {{\psi _i}} \right)SM_i\left[ \begin{array}{l}
{{\dot \phi }_i}\\
{{\dot \theta }_i}
\end{array} \right] - \left[ \begin{array}{l}
x_{id}^{\left( 3 \right)}\\
y_{id}^{\left( 3 \right)}
\end{array} \right]\\
&=\Xi_{i,1}+ {{\bar u}_{i,1}}{R_2}\left( {{\psi _i}} \right)SM_i\left[ \begin{array}{l}
{{\dot \phi }_i}\\
{{\dot \theta }_i}
\end{array} \right] - \left[ \begin{array}{l}
x_{id}^{\left( 3 \right)}\\
y_{id}^{\left( 3 \right)}
\end{array} \right],
\end{split}
\end{equation}
\end{small}
where
\begin{equation}\label{MMm}
\begin{split}
{{\dot {\bar u}}_{i,1}} &= z_{id}^{\left( 3 \right)} - {k_{3,z}}\left( {1 - {{\tanh }^2}{{\dot e}_{i,z}}} \right){{\ddot e}_{i,z}}\\
&~- {k_{1,z}}\left[ {1 - {{\tanh }^2}\left( {{{\dot e}_{i,z}} + {k_{2,z}}{e_{i,z}}} \right)} \right]\left( {{{\ddot e}_{i,z}} + {k_{2,z}}{{\dot e}_{i,z}}} \right),\\
\frac{{\mathrm{d}{R_2}\left( {{\psi _i}} \right)}}{{\mathrm{d}t}} &= {{\dot \psi }_i}{R_2}\left( {\frac{\pi }{2}} \right){R_2}\left( {{\psi _i}} \right),\\
M_i &= \left[ {\begin{array}{*{20}{c}}
0&{\sec^2\theta_{i}}\\
{\sec^2\phi_i \sec\theta_i}&{\tan\phi_i\tan\theta_i\sec\theta_i}
\end{array}} \right].
\end{split}
\end{equation}
Via some direct computations, $[e_{i,x}^{(4)},e_{i,y}^{(4)}]^T$ can be given by
\begin{equation}\label{sdsds}
\begin{split}
\left[ \begin{array}{l}
e_{i,x}^{\left( 4 \right)}\\
e_{i,y}^{\left( 4 \right)}
\end{array} \right]& =\Xi_{i,2}+ {{\bar u}_{i,1}}{R_2}\left( {{\psi _i}} \right)SM_i\left[ \begin{array}{l}
{u_{i,2}}\\
{u_{i,3}}
\end{array} \right] - \left[ \begin{array}{l}
x_{id}^{\left( 4 \right)}\\
y_{id}^{\left( 4 \right)}
\end{array} \right],
\end{split}
\end{equation}
with derivations shown in \eqref{sdsdaas} next page.
\begin{figure*}[!t]
\begin{equation}\label{sdsdaas}
\begin{split}
\left[ \begin{array}{l}
e_{i,x}^{\left( 4 \right)}\\
e_{i,y}^{\left( 4 \right)}
\end{array} \right]& = \underbrace{{{\dot \Xi }_{i,1}} +\left[ {{\dot {\bar u}}_{i,1}}{R_2}\left( {{\psi _i}} \right)SM_i + {{\bar u}_{i,1}}\frac{{{\rm{d}}{R_2}\left( {{\psi _i}} \right)}}{{{\rm{d}}t}}SM_i+ {{\bar u}_{i,1}}{R_2}\left( {{\psi _i}} \right)S\frac{{{\rm{d}}M_i}}{{{\rm{d}}t}}\right]\left[ \begin{array}{l}
{{\dot \phi }_i}\\
{{\dot \theta }_i}
\end{array} \right]}_{\Xi_{i,2}}+ {{\bar u}_{i,1}}{R_2}\left( {{\psi _i}} \right)SM_i\left[ \begin{array}{l}
{u_{i,2}}\\
{u_{i,3}}
\end{array} \right] - \left[ \begin{array}{l}
x_{id}^{\left( 4 \right)}\\
y_{id}^{\left( 4 \right)}
\end{array} \right],\\
{{\dot \Xi }_{i,1}} &= \left[ {{{\ddot {\bar u}}_{i,1}}{R_2}\left( {{\psi _i}} \right) + 2{{\dot {\bar u}}_{i,1}}\frac{{\mathrm{d}{R_2}\left( {{\psi _i}} \right)}}{{\mathrm{d}t}} + {{\bar u}_{i,1}}\frac{{{\mathrm{d}^2}{R_2}\left( {{\psi _i}} \right)}}{{\mathrm{d}{t^2}}}} \right]S\left[ \begin{array}{l}
\tan {\theta _i}\\
\displaystyle\frac{{\tan {\phi _i}}}{{\cos {\theta _i}}}
\end{array} \right] + \left( {{{\dot {\bar u}}_{i,1}}{R_2}\left( {{\psi _i}} \right) + {{\bar u}_{i,1}}\frac{{\mathrm{d}{R_2}\left( {{\psi _i}} \right)}}{{\mathrm{d}t}}} \right)SM_i\left[ \begin{array}{l}
{{\dot \phi }_i}\\
{{\dot \theta }_i}
\end{array} \right],\\
{{\ddot {\bar u}}_{i,1}} &= z_{id}^{\left( 4 \right)} + 2{k_{1,z}}\left[ {1 - {{\tanh }^2}\left( {{{\dot e}_{i,z}} + {k_{2,z}}{e_{i,z}}} \right)} \right]{\left( {{{\ddot e}_{i,z}} + {k_{2,z}}{{\dot e}_{i,z}}} \right)^2}\tanh \left( {{{\dot e}_{i,z}} + {k_{2,z}}{e_{i,z}}} \right) \\
&~- {k_{1,z}}\left[ {1 - {{\tanh }^2}\left( {{{\dot e}_{i,z}} + {k_{2,z}}{e_{i,z}}} \right)} \right]\left( {e_{i,z}^{\left( 3 \right)} + {k_{2,z}}{{\ddot e}_{i,z}}} \right)
+ 2{k_{3,z}}\left( {1 - {{\tanh }^2}{{\dot e}_{i,z}}} \right)\ddot e_{i,z}^2\tanh {{\dot e}_{i,z}} - {k_{3,z}}\left( {1 - {{\tanh }^2}{{\dot e}_{i,z}}} \right)e_{i,z}^{\left( 3 \right)},\\
e_{i,z}^{\left( 3 \right)} &=  - {k_{1,z}}\left[ {1 - {{\tanh }^2}\left( {{{\dot e}_{i,z}} + {k_{2,z}}{e_{i,z}}} \right)} \right]\left( {{{\ddot e}_{i,z}} + {k_{2,z}}{{\dot e}_{i,z}}} \right) - {k_{3,z}}\left( {1 - {{\tanh }^2}{{\dot e}_{i,z}}} \right){\ddot e_{i,z}},\\
\frac{{{\rm{d}}R_2^2\left( {{\psi _i}} \right)}}{{{\rm{d}}{t^2}}} &= {{\ddot \psi }_i}{R_2}\left( {\frac{\pi }{2}} \right){R_2}\left( {{\psi _i}} \right) + \dot \psi _i^2R_2^2\left( {\frac{\pi }{2}} \right){R_2}\left( {{\psi _i}} \right),\\
\frac{{\mathrm{d}M_i}}{{\mathrm{d}t}} &= \left[ {\begin{array}{*{20}{c}}
0&{2{{\dot \theta }_i}{{\sec }^2}{\theta _i}\tan {\theta _i}}\\
{2{{\dot \phi }_i}{{\sec }^2}{\phi _i}\tan {\phi _i}\sec {\phi _i} + {{\dot \theta }_i}{{\sec }^2}{\phi _i}\tan {\theta _i}\sec {\theta _i}}&{{{\dot \phi }_i}{{\sec }^2}{\phi _i}\tan {\theta _i}\sec {\theta _i} + {{\dot \theta }_i}\tan {\phi _i}\left( {{{\sec }^3}{\theta _i} + \tan {\theta _i}\sec {\theta _i}} \right)}
\end{array}} \right].
\end{split}
\end{equation}
\hrulefill
\end{figure*}
\\
Then, the error states $(\ref{sdsds})$ motivates
\begin{equation}\label{u23i}
\left[ \begin{array}{l}
{u_{i,2}}\\
{u_{i,3}}
\end{array} \right] = \frac{{{M_i^{ - 1}}{S^{ - 1}}R_2^{ - 1}\left( {{\psi _i}} \right)}}{{{{\bar u}_{i,1}}}} \{{\left[ \begin{array}{l}
x_{id}^{\left( 4 \right)}\\
y_{id}^{\left( 4 \right)}
\end{array} \right] - {\Xi _{i,2}} + \left[ \begin{array}{l}
{u_{i,x}}\\
{u_{i,y}}
\end{array} \right]} \},
\end{equation}
where $[u_{i,x},u_{i,y}]^T$ is viewed as new control input, $S^{-1}=S$, $R_2(\psi_i)=R_2^T(\psi_i)$ and
\begin{equation}\label{MatrixMMM}
{M_i^{ - 1}} = \left[ {\begin{array}{*{20}{c}}
{-0.25\sin 2{\phi _i}\sin 2{\theta _i}}&{{{\cos }^2}{\phi _i}\cos {\theta _i}}\\
{{{\cos }^2}{\theta _i}}&0
\end{array}} \right].
\end{equation}
Via substituting $(\ref{u23i})$ into $(\ref{sdsds})$, we obtain the linearized form of $[e_{i,x},e_{i,y}]^T$-dynamics by
\begin{equation}\label{fbli}
\left\{ \begin{array}{l}
e_{i,x}^{\left( 4 \right)} = {u_{i,x}}\\
e_{i,y}^{\left( 4 \right)} = {u_{i,y}}
\end{array} \right..
\end{equation}
Concerning the linear form of $(\ref{fbli})$, there are various strategies to formulate $[u_{i,x},u_{i,y}]^T$, such as nested linear technique and nested saturation function method \cite{RN911,RN912}. For brief, we propose a linear control law by
\begin{equation}\label{lincotrn}
\left\{ \begin{split}
{u_{i,x}} &=  - {k_{1,x}}{e_{i,x}} - {k_{2,x}}{{\dot e}_{i,x}} - {k_{3,x}}{{\ddot e}_{i,x}} - {k_{4,x}}e_{i,x}^{\left( 3 \right)}\\
{u_{i,y}} &=  - {k_{1,y}}{e_{i,y}} - {k_{2,y}}{{\dot e}_{i,y}} - {k_{3,y}}{{\ddot e}_{i,y}} - {k_{4,y}}e_{i,y}^{\left( 3 \right)}
\end{split} \right.
\end{equation}
where the control gains $k_{1,x},k_{2,x},k_{3,x},k_{4,x},k_{1,y},k_{2,y},k_{3,y},k_{4,y}$ should be selected so that the linearized system $(\ref{fbli})$ is exponentially stable. Till now, the design of local tracking control law is completed. The following theorem validates the local control law.
\begin{theorem}
Given any fourth-order differentiable reference trajectory $p_{id}=[x_{id},y_{id},z_{id}]^T$ with bounded derivatives and $|\ddot{z}_{id}(t)|<g$, the application of control laws $(\ref{yawctrl})(\ref{ui14})(\ref{u23i})(\ref{lincotrn})$ on the quad-rotor model $(\ref{SimpModel})$ ensures that
\begin{equation}\label{thsx}
\mathop {\lim }\limits_{t \to  + \infty } \left[ \begin{array}{l}
{x_i} - {x_{id}}\\
{y_i} - {y_{id}}\\
{z_i} - {z_{id}}
\end{array} \right] = \left[ \begin{array}{l}
{d_{i,x}}\\
{d_{i,y}}\\
{d_{i,z}}
\end{array} \right],\mathop {\lim }\limits_{t \to  + \infty } {\psi _i} = 0.
\end{equation}
\end{theorem}
\begin{proof}
By the proposed control laws $(\ref{yawctrl})(\ref{ui14})(\ref{u23i})(\ref{lincotrn})$, the $[e_{i,x},e_{i,y},e_{i,z}]$-dynamics and the $\psi_i$-dynamics can be written as follows,
\begin{equation}\label{clol}
\left\{
\begin{split}
e_{i,x}^{\left( 4 \right)} &=  - {k_{1,x}}{e_{i,x}} - {k_{2,x}}{{\dot e}_{i,x}} - {k_{3,x}}{{\ddot e}_{i,x}} - {k_{4,x}}e_{i,x}^{\left( 3 \right)}\\
e_{i,y}^{\left( 4 \right)} &=  - {k_{1,y}}{e_{i,y}} - {k_{2,y}}{{\dot e}_{i,y}} - {k_{3,y}}{{\ddot e}_{i,y}} - {k_{4,y}}e_{i,y}^{\left( 3 \right)}\\
{{\ddot e}_{i,z}} &=  - {k_{1,z}}\tanh \left( {{{\dot e}_{i,z}} + {k_{2,z}}{e_{i,z}}} \right) - {k_{3,z}}\tanh{{\dot e}_{i,z}}\\
{{\ddot \psi }_i} &=  - {k_{1,\psi }}{\psi _i} - {k_{2,\psi }}{{\dot \psi }_i}
\end{split}\right.
\end{equation}
Based on $(\ref{clol})$ and Lemma 1 as well as classical linear stability theorem, it is direct to obtain $(\ref{thsx})$. Moreover, the latitudinal and longitudinal states $e_{i,x},\dot{e}_{i,x},\ddot{e}_{i,x},e_{i,x}^{(3)},e_{i,y},\dot{e}_{i,y},\ddot{e}_{i,y}$ and $e_{i,y}^{(3)}$ globally uniformly exponentially converge to zero. The altitude errors $e_{i,z},\dot{e}_{i,z}$ and $\ddot{e}_{i,z}$ are globally uniformly asymptotically convergent.
\end{proof}
Observing equations $(\ref{ui14})-(\ref{u23i})$, one may find out that it is essential to bound $\phi_i$ and $\theta_i$ in $(-\pi/2,\pi/2)$ regarding the designs of $(\ref{ui14})$ and $(\ref{u23i})$. To demonstrate the condition under which $\phi_i,\theta_i\in$$(-\pi/2,\pi/2)$ establishes, the proposition below is introduced.
\begin{proposition}
Given any virtual reference trajectory $p_{id}=[x_{id},y_{id},z_{id}]^T$ with bounded derivatives and satisfying $|\ddot{z}_{id}(t)|<g$, applying any control law $[u_{i,x},u_{i,y}]^T$ capable of achieving $(\ddot{e}_{i,x},\ddot{e}_{i,y})\in L_\infty$, and $u_{i,1}$ defined in (11) and any $u_{i,4}\in L_\infty$ on the quad-rotor $(\ref{SimpModel})$ ensures that $\phi_i(t),\theta_i(t)\in(-\pi/2,\pi/2), \forall t \geq 0$.
\end{proposition}
\begin{proof}
By $(\ref{ddexy})$, one has
\begin{equation}\label{sszz}
\left[ \begin{array}{l}
\tan {\theta _i}\\
\displaystyle\frac{{\tan {\phi _i}}}{{\cos {\theta _i}}}
\end{array} \right] = \frac{{R_2^{ - 1}\left( {{\psi _i}} \right){S^{ - 1}}}}{{{{\bar u}_{i,1}}}}\left[ {\left[ \begin{array}{l}
{{\ddot e}_{i,x}}\\
{{\ddot e}_{i,y}}
\end{array} \right] + \left[ \begin{array}{l}
{{\ddot x}_{id}}\\
{{\ddot y}_{id}}
\end{array} \right]} \right],
\end{equation}
which, together with $\ddot{e}_{i,x},\ddot{e}_{i,y}\in L_\infty$ drawn from Theorem 1 and $\bar{u}_{i,1}>0$ due to $k_{1,z}+k_{3,z}<g-|\ddot{z}_{id}|$, shows that $\tan\theta_i\in L_\infty$ and $\tan\phi_i\in L_\infty$. Hence, $\phi_i(t),\theta_i(t)\in(-\pi/2,\pi/2),\forall t \geq 0$.
\end{proof}
The proposition 1 illustrates the sufficient condition to avoid singularity when calculating the control laws $(\ref{ui14})(\ref{u23i})(\ref{lincotrn})$, that is, basically, ensuring $\ddot{e}_{i,x},\ddot{e}_{i,y}\in L_\infty$.
\begin{remark}
By $(\ref{sszz})$, further computation can obtain the roll and pitch angles of steady state, with trival tracking errors, as follows,
\begin{equation}\label{zxl2}
\left\{
\begin{split}
{\phi _i} &= \arctan  (- \frac{{{{\ddot y}_{id}}}}{{\sqrt {\ddot x_{id}^2 + {{\left( {g + {{\ddot z}_{id}}} \right)}^2}} }})\\
{\theta _i} &= \arctan \frac{{{{\ddot x}_{id}}}}{{g + {{\ddot z}_{id}}}}
\end{split}\right..
\end{equation}
\end{remark}
\begin{remark}
Note that the control inputs $(u_{i,2},u_{i,3})$ are directly designed with position error states rather than steering the roll and pitch to track their reference signals as traditional inner-outer loop methods do. One of such reference signals refers to the virtual roll and pitch angles $(\phi_{id},\theta_{id})$ solved according to
\begin{equation}
\begin{array}{l}
\left[ \begin{array}{l}
\tan {\theta _{id}}\\
\displaystyle\frac{{\tan {\phi _{id}}}}{{\cos {\theta _{id}}}}
\end{array} \right] = \displaystyle\frac{{R_2^{ - 1}\left( {{\psi _i}} \right){S^{ - 1}}}}{{{{\bar u}_{i,1}}}}\{ \left[ \begin{array}{l}
{{\ddot e}_{i,x}}\\
{{\ddot e}_{i,y}}
\end{array} \right] + \left[ \begin{array}{l}
{{\ddot x}_{id}}\\
{{\ddot y}_{id}}
\end{array} \right]\\
\;\;\;\;\;\;\;\;\;\;\;\;\;\;\;\;\;\;\;\;\;\;\;\;\;\;\;\;\;\;\;\;\;\;\;\;\;\;\; + {b_1}\left[ \begin{array}{l}
{e_{i,x}}\\
{e_{i,y}}
\end{array} \right] + {b_2}\left[ \begin{array}{l}
{{\dot e}_{i,x}}\\
{{\dot e}_{i,y}}
\end{array} \right]\}
\end{array}
\end{equation}
with suitable $b_1,b_2>0$. Steering $(\phi_i,\theta_i)$ to $(\phi_{id},\theta_{id})$ can then lead to $\left[\ddot{e}_{i,x},\ddot{e}_{i,y}\right]^T=-b_1\left[{e}_{i,x},{e}_{i,y}\right]^T-b_2\left[\dot{e}_{i,x},\dot{e}_{i,y}\right]^T$ and the position error $[e_{i,x},e_{i,y}]^T$ would converge to zero.
\end{remark}
\begin{remark}
Different form traditional feedback linearization results having fourteen states to be regulated \cite{RN863,RN858,RN874,RN855,RN848,RN843,RN856}, there are only twelve control states to be regulated in our design.
\end{remark}
\begin{remark}
By $(\ref{fbli})$, the latitudinal and longitudinal position errors are converted into two fourth-order integrators via viewing $(u_{i,x},u_{i,y})$ as control inputs. It is therefore able to design $(u_{i,x},u_{i,y})$ with the help of classical linear techniques, such as saturated method \cite{RN911,RN912}.
\end{remark}
\begin{remark}
The application of control law $(\ref{yawctrl})(\ref{ui14})(\ref{u23i})(\ref{lincotrn})$ on a real quad-rotor should consider more practical scenarios such as available thrust, mass and reference acceleration as well as other requirements. This note does not take these factors into account to make the main idea be presented in a concise manner.
\end{remark}
\subsection{Distributed Observer (Virtual Reference Trajectory)}
\noindent To achieve formation via a local tracking control law, as stated above, the virtual reference trajectory $p_{id}=[x_{id},y_{id},z_{id}]^T$ should have bounded derivatives and satisfy $\mathop {\lim }\limits_{t \to  + \infty } {\left[ {{x_{id}},{y_{id}},{z_{id}}} \right]^T} = {\left[ {{x_0},{y_0},{z_0}} \right]^T}$ and $\left| {{{\ddot z}_{id}}\left( t \right)} \right| < g$. Two problems obstruct the design of such reference trajectory, that is, only partial quad-rotors can know the leader's states and only neighboring interaction is available. Concerning these problems and the fact that $z_{id}$ needs an acceleration less than $g$ while there is no such restriction about $[x_{id},y_{id}]^T$, the designs of $(x_{id},y_{id})$ and $z_{id}$ are separately presented in two lemmas below.\\
Let $\zeta_{id}=[x_{id},y_{id}]^T$ and $\zeta_0=[x_0,y_0]^T$ and extend the second-order observer reported in our previous work \cite{RNylx}, propose the following lemma.
\begin{lemma}
Given Assumptions 2-3, the fourth-order dynamics described by
\begin{equation}\label{obsp}
\zeta _{id}^{\left( 4 \right)} =  - {g_3}\zeta _{id}^{\left( 3 \right)} - {g_2}{{\ddot \zeta }_{id}} - {g_1}{{\dot \zeta }_{id}} - {g_4}{c_i} - {Q_i}{c_i}
\end{equation}
with
\begin{equation}\label{papra}
\begin{split}
{c_i} &= {\left[ {{c_{i,x}},{c_{i,y}}} \right]^T}\\
 &= \sum\limits_{j = 1}^n {{a_{ij}}\left( {\zeta_{id}^{\left( 3 \right)} - \zeta_{jd}^{\left( 3 \right)}} \right)}  + {a_{i0}}\left( {\zeta_{id}^{\left( 3 \right)} - \zeta_0^{\left( 3 \right)}} \right)\\
 &~+ {g_3}\sum\limits_{j = 1}^n {{a_{ij}}\left( {{{\ddot \zeta}_{id}} - {{\ddot \zeta}_{jd}}} \right)}  + {g_3}{a_{i0}}\left( {{{\ddot \zeta}_{id}} - {{\ddot \zeta}_0}} \right)\\
 &~+ {g_2}\sum\limits_{j = 1}^n {{a_{ij}}\left( {{{\dot \zeta}_{id}} - {{\dot \zeta}_{jd}}} \right)}  + {g_2}{a_{i0}}\left( {{{\dot \zeta}_{id}} - {{\dot \zeta}_0}} \right)\\
 &~+ {g_1}\sum\limits_{j = 1}^n {{a_{ij}}\left( {{\zeta_{id}} - {\zeta_{jd}}} \right)}  + {g_2}{a_{i0}}\left( {{\zeta_{id}} - {\zeta_0}} \right)\\
{Q_i} &= \mathrm{diag}\left\{ {\left[ {\frac{{{g_{5,x}}}}{{\left| {{c_{i,x}}} \right| + \gamma {e^{ - \lambda t}}}},\frac{{{g_{5,y}}}}{{\left| {{c_{i,y}}} \right| + \gamma {e^{ - \lambda t}}}}} \right]} \right\},
\end{split}
\end{equation}
and gain selections
\begin{equation}\label{gains}
\begin{split}
g_2&>0,g_3>0,g_2g_3>g_1>0,g_4>0,\gamma>0,\lambda>0,\\
{g_{5,x}}& \ge \sigma_{0,x}:=\mathop {\sup }\limits_{t \ge 0}\left| {x_0^{(4)}+{g_3}x_0^{\left( 3 \right)} + {g_2}{{\ddot x}_0} + {g_1}{{\dot x}_0}} \right|,\\
{g_{5,y}}& \ge \sigma_{0,y}:=\mathop {\sup }\limits_{t \ge 0}\left| {y_0^{(4)}+{g_3}y_0^{\left( 3 \right)} + {g_2}{{\ddot y}_0} + {g_1}{{\dot y}_0}} \right|,
\end{split}
\end{equation}
ensures that $\zeta_{id},\dot{\zeta}_{id},\ddot{\zeta}_{id}$ and $\zeta^{(3)}_{id}$ are bounded, and
\begin{equation}\label{goa}
\begin{split}
\mathop {\lim }\limits_{t \to \infty } {\zeta_{id}} &= {\zeta_0},\mathop {\lim }\limits_{t \to \infty } {{\dot \zeta}_{id}} = {{\dot \zeta}_0},
\mathop {\lim }\limits_{t \to \infty } {{\ddot \zeta}_{id}} = {{\dot \zeta}_0},\mathop {\lim }\limits_{t \to \infty } \zeta_{id}^{\left( 3 \right)} = \zeta_0^{\left( 3 \right)}.
\end{split}
\end{equation}
\end{lemma}
\begin{proof}
See Appendix A.
\end{proof}
For $z_{id}$, we propose the following lemma,
\begin{lemma}
The second-order dynamics
\begin{equation}\label{zsds2}
{{\ddot z}_{id}} =  - {h_1}\tanh \left( {{{\dot z}_{id}} + {h_2}\left( {{z_{id}} - {z_{ia}}} \right)} \right) - {h_3}\tanh {{\dot z}_{id}}
\end{equation}
driven by
\begin{equation}\label{zidid}
\left\{ \begin{split}
{{\ddot z}_{ia}} &=  - {h_4}\left( {{z_{ia}} - {z_{ib}}} \right) - {h_5} {{\dot z}_{ia}}\\
{{\dot z}_{ib}} &=  - {h_6}\sum\limits_{j = 1}^n {{a_{ij}}\left( {{z_{ib}} - {z_{jb}}} \right) - {h_6}{a_{i0}}\left( {{z_{ib}} - {z_0}} \right)}
\end{split} \right.
\end{equation}
with $h_1,h_2,h_3,h_4,h_5,h_6>0$ and $h_1+h_3<g$, ensures that
\begin{enumerate}
  \item $|\ddot{z}_{id}(t)|<g$;
  \item $\mathop {\lim }\limits_{t \to  + \infty } {z_{id}} = {z_0}$.
\end{enumerate}
\end{lemma}
\begin{proof}
The first claim can be verified directly by selecting gains so that $h_1+h_3<g$. To prove the second claim, we prove $z_{ib}\to z_0, z_{ia}\to z_0$ and $z_{id}\to z_0$ sequently. As reported in many classical literatures, see \cite{RN272} for example, it is direct to conclude that $z_{ib}-z_0$ and $\dot{z}_{ib}$ are GES with decaying rate being much related to the smallest eigenvalue of the matrix $\mathcal{H}$ \cite{RN848}. Therefore, by $\ddot{z}_{ia}=-h_4(z_{ia}-z_0)-h_5\dot{z}_{ia}+h_4(z_{ib}-z_0)$ derived from $(\ref{zidid})$, $z_{ia}-z_0$ converges to zero globally exponentially. For the second-order dynamics $(\ref{zsds2})$, define errors
$\varepsilon_{i,1}=z_{id}-z_0,\varepsilon_{i,2}=\dot{\varepsilon}_{i,1}$ with time-derivatives given by
\begin{equation}\label{deps}
\left\{ \begin{split}
{{\dot \varepsilon }_{i,1}} &= {\varepsilon _{i,2}}\\
{{\dot \varepsilon }_{i,2}} &=  - {h_1}\tanh \left( {{\varepsilon _{i,2}} + {h_2}{\varepsilon _{i,1}}} \right) - {h_3}\tanh {\varepsilon _{i,2}}+\Delta_i
\end{split} \right.,
\end{equation}
where $\Delta_i=- {h_1}(\tanh \left( {{\varepsilon _{i,2}} + {h_2}\left( {{z_{id}} - {z_{ia}}} \right)} \right)- \tanh ( {\varepsilon _{i,2}}$\\$+ {h_2}\left( {{z_{id}} - {z_0}} \right))$. By the mean value theorem, we have $|\Delta_i|<h_1|z_{ia}-z_0|$. Hence, $\Delta_i$ globally exponentially converges to zero. For the nominal part $\dot{\varepsilon}_{i,1}=\varepsilon_{i,2},\dot{\varepsilon}_{i,2}=-h_1\tanh(\varepsilon_{i,2}+h_2\varepsilon_{i,1})-h_3\tanh(\varepsilon_{i,2})$, one can choose a positive definite function by $V = {h_1}\ln \cosh \left( {{\varepsilon _{i,2}} + {h_2}{\varepsilon _{i,1}}} \right) + {h_3}\ln \cosh {\varepsilon _{i,2}} + \frac{1}{2}{h_2}\varepsilon _{i,2}^2$ \cite{RN842}, whose derivative is $\dot V =  - {\left[ {{h_1}\tanh \left( {{\varepsilon _{i,2}} + {h_2}{\varepsilon _{i,1}}} \right) + {h_3}\tanh {\varepsilon _{i,2}}} \right]^2} - {h_2}{h_3}{\varepsilon _{i,2}}\tanh {\varepsilon _{i,2}}<0$. Therefore, the nominal system associated with $(\ref{deps})$ is globally asymptotically stable, which, together with the cascade theory \cite{RN19,RN80} and the fact that $\Delta_i$ is GES, shows that $[\varepsilon_{i,1},\varepsilon_{i,2}]^T$ globally asymptotically converges to zero. The claim $\mathop {\lim }\limits_{t \to  + \infty } {z_{id}} = {z_0}$ follows.
\end{proof}
The observers proposed in Lemma 2 and Lemma 3 solve the problem that the leader's states are not available to all quad-rotors. They act as interactions among quad-rotors, and can be generated by on-board computer and transmitted by wireless modules. Moreover, the derivatives of the proposed observer are bounded up to the fourth order and satisfy $\ddot{z}_{id}<g$. It is therefore admissible to view the observers $(\ref{obsp})$$(\ref{zsds2})$ as virtual reference trajectory for each quad-rotor and apply the local control law stated in previous subsection.
\subsection{Brief discussion}
The theorem below concludes the formation scheme briefly.
\begin{theorem}
Given Assumptions 1-3, applying the local tracking control laws $(\ref{yawctrl})(\ref{ui14})(\ref{u23i})(\ref{lincotrn})$ and distributed observers $(\ref{obsp})(\ref{zsds2})$ on a team of quad-rotors described by $(\ref{SimpModel})$ achieves
\begin{equation}\label{sdszx}
\mathop {\lim }\limits_{t \to \infty } \left[ \begin{array}{l}
{x_i} - {x_0}\\
{y_i} - {y_0}\\
{z_i} - {z_0}
\end{array} \right] = \left[ \begin{array}{l}
{d_{i,x}}\\
{d_{i,y}}\\
{d_{i,z}}
\end{array} \right],\forall i \in \mathcal{N}.
\end{equation}
\end{theorem}
\begin{proof}
The actual formation error $p_i-p_0-\Delta_i$ satisfies
\begin{equation}\label{sadaz}
\begin{split}
\left\| {{p_i} - {p_0} - {\Delta _i}} \right\| &= \left\| {{p_i} - {p_{id}} - {\Delta _i} + {p_{id}} - {p_0}} \right\|\\
 &\le \left\| {{p_i} - {p_{id}} - {\Delta _i}} \right\| + \left\| {{p_{id}} - {p_0}} \right\|,
\end{split}
\end{equation}
which, together with Theorem 1, Lemma 2 and Lemma 3, shows that $\mathop {\lim }\limits_{t \to \infty } {p_i} - {p_0} - {\Delta _i} = {\left[ {0,0,0} \right]^T}$. The claim $(\ref{sdszx})$ follows.
\end{proof}
\begin{remark}
Inspired by the analysis above, one can solve other cooperative problems of quad-rotor systems (such as time-varying formation and communication delay) via modifying associated control protocols developed for linear integrators into virtual reference trajectories meeting the requirements of our local controller.
\end{remark}
\begin{remark}
The formation scheme is distributed since only neighboring communication is available. It is therefore direct to add cooperative quad-rotors as required.
\end{remark}
\section{Numerical Simulation}
To validate the proposed distributed formation control algorithm we employ a scenario for simulation, concerning four quad-rotors tracking a leader while performing a fixed square pattern. Without loss of generality, the undirected interaction network is described by the figure below.
\begin{figure}[H]
  \centering
\setlength{\abovecaptionskip}{0.cm}
\includegraphics[width=0.20\columnwidth]{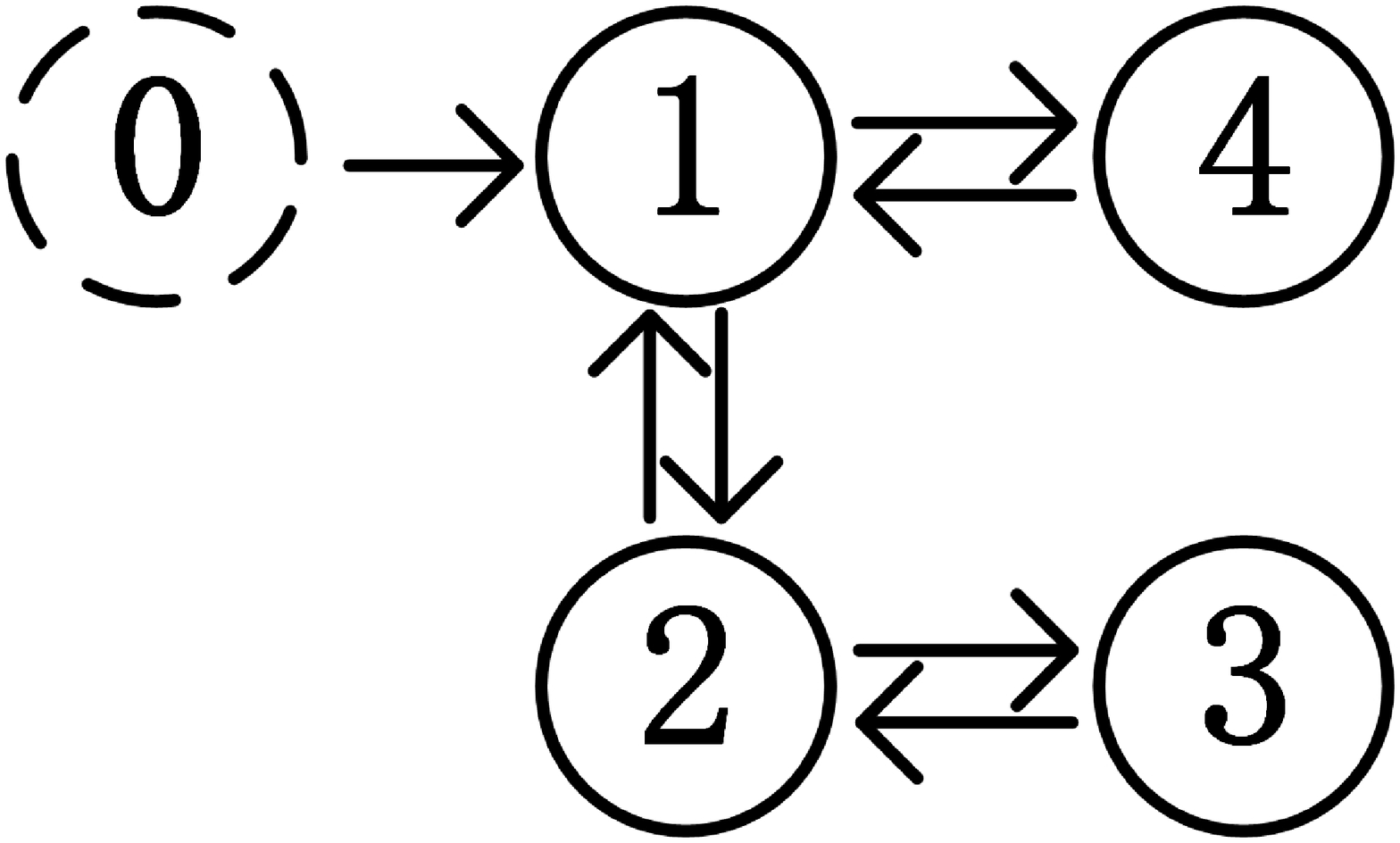}
\caption{The interaction network.}
\label{fig1}
\end{figure}
To form a fixed square pattern, define four constant vectors by $\Delta_1 =[20,20,0]^T, \Delta_2=[-20,20,0]^T,\Delta_3 =[-20,-20,0]^T,\Delta_4=[20,-20,0]^T$.
Two cases are considered as follows,
\begin{itemize}
  \item \textbf{Case1:}~~${p_0}\left( t \right) = {\left[ {100\sin 0.1t, - 100\cos 0.1t,100} \right]^T}$,
  \item \textbf{Case2:}~~${p_0}\left( t \right) = {\left[ {0,0,50} \right]^T}$.
\end{itemize}
For convenience, we select the identical initial states by$
{x_1}\left( 0 \right) =  - 10,{y_1}\left( 0 \right) = 12,{z_1}\left( 0 \right) = 0,{\phi _1}\left( 0 \right) = 0,{\theta _1}\left( 0 \right) = 0,{\psi _1}\left( 0 \right) = \pi /8,{x_2}\left( 0 \right) = 40,{y_2}\left( 0 \right) =  - 12,{z_2}\left( 0 \right) = 5,{\phi _2}\left( 0 \right) = 0,{\theta _2}\left( 0 \right) = 0,{\psi _2}\left( 0 \right) = \pi /2,{x_3}\left( 0 \right) = 20,{y_3}\left( 0 \right) = 10,{z_3}\left( 0 \right) = 6,{\phi _3}\left( 0 \right) = 0,{\theta _3}\left( 0 \right) = 0,{\psi _3}\left( 0 \right) = \pi ,{x_4}\left( 0 \right) =  - 20,{y_4}\left( 0 \right) = 45,{z_4}\left( 0 \right) = 7,{\phi _4}\left( 0 \right) = 0,{\theta _4}\left( 0 \right) = 0,{\psi _4}\left( 0 \right) = \pi /5.$ The control gains are also selected to be identically for two cases as ${k_{1,z}} = 1,{k_{2,z}} = 0.5,{k_{3,z}} = 0.5,
{k_{1,x}} = 0.2,{k_{2,x}} = 1.6,{k_{3,x}} = 3.6,{k_{4,x}} = 3.2,{k_{1,y}} = 0.2,{k_{2,y}} = 1.6,{k_{3,y}} = 3.6,{k_{4,y}} = 3.2,{k_{1,\psi }} = 0.5,{k_{2,\psi }} = 0.5,
{g_1} = 0.125,{g_2} = 0.75,{g_3} = 0.85,{g_4} = 0.1,
{g_{5,x}} = 2.1,{g_{5,y}} = 2.1,\gamma  = 15,\lambda  = 0.1,
{h_1} = 0.5,{h_2} = 0.5,{h_3} = 0.5,
{h_4} = 0.5,{h_5} = 0.5,{h_6} = 1.$ To start the simulation, the initial values of the virtual trajectory is chosen as $\zeta_{id}=[x_{id}(0),y_{id}(0)]^T=[x_i(0),y_i(0)]^T$ and $z_{id}(0)=z_i(0)$ with $i=\{1,2,3,4\}$ and their derivatives are supposed to be zero.\\
Three sub-figures are depicted for each case with sub-figure (a) being the geometric position paths of quad-rotors, sub-fugure (b) being the norm of formation error $p_i-p_0-\Delta_i$ and sub-figure (c) being the attitude angles. The simulation results of these two cases are shown in Figure 4 and Figure 5 respectively. It can be seen from sub-figure (a) of both two cases that the follower quad-rotors form the square pattern while tracking the leader with the predefined position displacements. The sub-figure (b) shows that the formation error norm is asymptotically convergent. The attitudes of all followers are kept in reasonable ranges that can be demonstrated by sub-figure (c). All of the simulation results illustrate the effectiveness of the proposed formation algorithm.
\section{Conclusion}
This note solves the leader-follower formation problem for multiple quad-rotors. The whole formation control scheme involves two parts, namely, a local tracking control law and a distributed observer. Given a smooth reference trajectory with bounded derivatives, a novel local tracking control law is proposed with the help of non-regular feedback linearization approach. In view of the fact that the leader's states are not available to all followers, we propose a distributed observer and feed it into the local tracking control law via viewing it as virtual reference trajectory. As for future research, we will take into account more practical problem associated quad-rotor formation such as communication failure, wind turbulence and inter-agent collision avoidance.
\begin{figure}[H]
\centering     
\subfigure[Geometric position paths(*:Start point).]{\label{fig:a}\includegraphics[width=0.8\columnwidth]{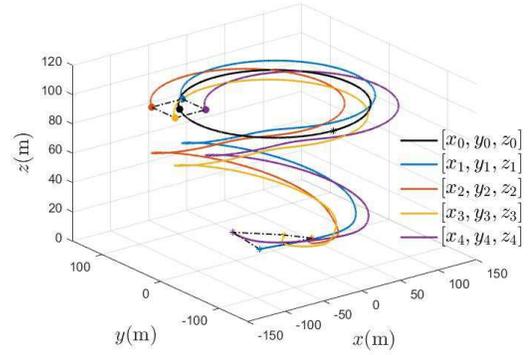}}
\subfigure[The norm of formation errors.]{\label{fig:b}\includegraphics[width=.80\columnwidth]{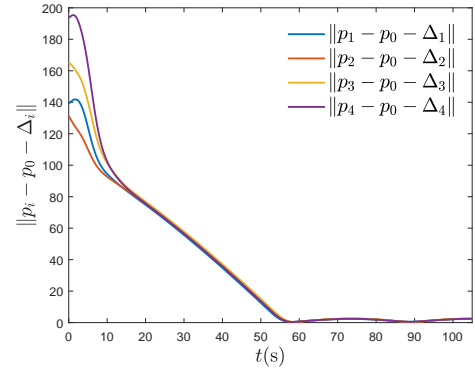} }\\
\subfigure[The attitude angles of each quad-rotor(Unit:degree).]{\label{fig:b}\includegraphics[width=.90\columnwidth]{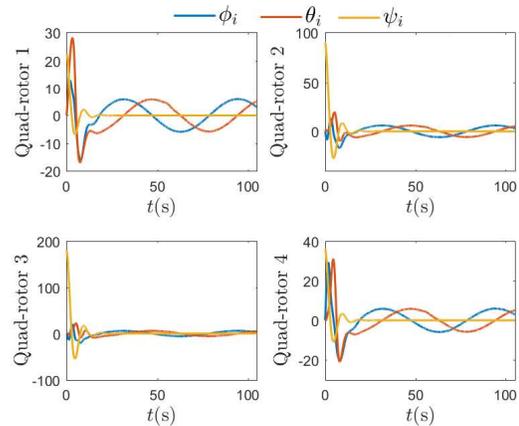} }\\
\caption{The simulation results of \textbf{Case 1}.}
\end{figure}
\begin{figure}[H]
\centering     
\subfigure[Geometric position paths(*:Start point).]{\label{fig:a}\includegraphics[width=0.8\columnwidth]{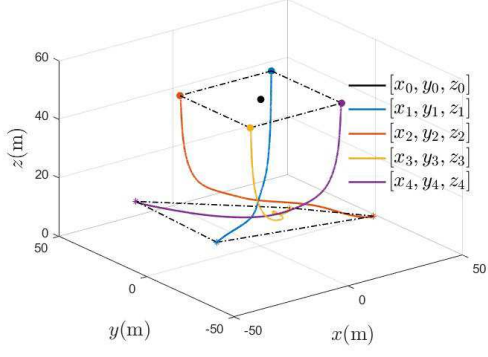}}
\subfigure[The norm of formation errors.]{\label{fig:b}\includegraphics[width=.8\columnwidth]{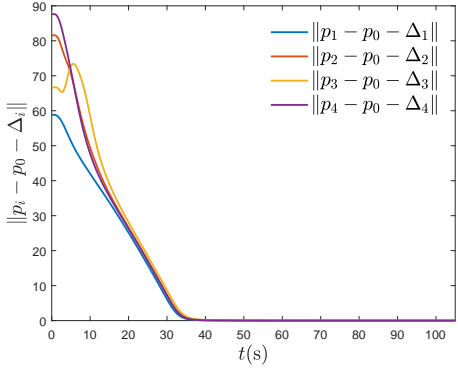} }\\
\subfigure[The attitude angles of each quad-rotor(Unit:degree).]{\label{fig:b}\includegraphics[width=.90\columnwidth]{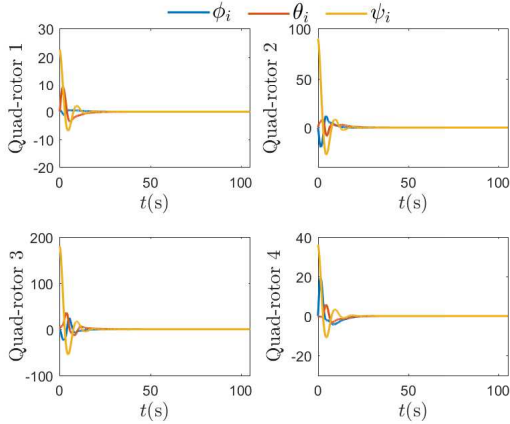} }\\
\caption{The simulation results of \textbf{Case 2}.}
\end{figure}
\section*{Acknowlegement}
{This work was supported by National Natural Science Foundation of China under Grant 61573034 and Grant 61327807.}
\section*{Appendix A: Proof of Lemma 2}
\begin{proof}
Define
\begin{equation}\label{xixi}
\begin{split}
\xi_i&=\zeta_{id}-\zeta_0\in\mathbb{R}^2,s_i=\xi_i^{(3)}+g_3\ddot{\xi}_i+g_2\dot{\xi}_i+g_1\xi_i\in\mathbb{R}^2,
\end{split}
\end{equation}
and note that $s_i=0$ leads to $\xi_i^{(3)}=-g_3\ddot{\xi}_i-g_2\dot{\xi}_i-g_1\xi_i$. As a result, by gains $g_2>0,g_3>0,g_2g_3>g_1>0$, the objective $(\ref{goa})$ can be realized. Hence, let us prove $s_i\to 0,\forall i \in \mathcal{N}$ and define some helpful $2n$-dimensional vectors by
\begin{equation}\label{cxi}
\begin{split}
   c & =[c_1^T,c_2^T,...,c_n^T]^T,\xi =[\xi_1^T,\xi_2^T,...,\xi_n^T]^T,s=[s_1^T,s_2^T,...,s_n^T]^T.
\end{split}
\end{equation}
Direct computations result with
\begin{equation}\label{ccssx}
\begin{split}
c &= \left( {\mathcal{H} \otimes {I_2}} \right)s,\\
s&=\xi^{(3)}+g_3\ddot{\xi}+g_2\dot{\xi}+g_1\xi,
\end{split}
\end{equation}
where '$\otimes$' denotes Kronecker product. By Assumption 3, the matrix $\mathcal{H}$ is symmetric, invertible and positive definite \cite{RN531}, and hence, so is $\mathcal{H}\otimes I_2$. The convergence $c\to 0_{2n}$ is equivalent with $s\to 0_{2n}$. Choose a positive definite function
\begin{equation}\label{Lyusa}
W=0.5c^T(\mathcal{H}\otimes I_2)^{-1}c,
\end{equation}
whose derivative can be obtained as $\dot{W}= c^T(\mathcal{H}\otimes I_2)\dot{c}=c^T\dot{s}$. By the fact
\begin{equation}\label{dotssss}
\begin{split}
\dot s &= {\zeta_{id}^{\left( 4 \right)}} - {1_n} \otimes \zeta _0^{\left( 4 \right)} + {g_3}\left( {\zeta _{id}^{\left( 3 \right)} - {1_n} \otimes \zeta _0^{\left( 3 \right)}} \right)+ {g_2}\left( {{{\ddot \zeta }_{id}} - {1_n} \otimes {{\ddot \zeta }_0}} \right)\\
&~ + {g_1}\left( {{{\dot \zeta }_{id}} - {1_n} \otimes {{\dot \zeta }_0}} \right)\\
 &=  - {g_4}c - Qc - {1_n} \otimes \left( {\zeta _0^{\left( 4 \right)} + {g_3}\zeta _0^{\left( 3 \right)} + {g_2}{{\ddot \zeta }_0} + {g_1}{{\dot \zeta }_0}} \right),
\end{split}
\end{equation}
with $Q=\textrm{diag}\{Q_1,Q_2,...,Q_n\}\in \mathbb{R}^{2n\times 2n}$, rewrite $\dot W$ by,
\begin{equation}\label{Lyause}
\begin{split}
\dot{W}&=-g_4c^Tc-c^TGc-c^T\mathbf{1}_n\otimes(\zeta_0^{(4)}+g_3 \zeta_0^{(3)}+g_2 \ddot{\zeta}_0+g_1 \dot{\zeta}_0)\\
        &=-g_4\|c\|^2 \\
       &~- \sum\limits_{i = 1}^n {\frac{{{g_{5,x}}c_{i,x}^2}}{{\left| {{c_{i,x}}} \right| + \gamma {e^{ - \lambda t}}}}}
       -\sum\limits_{i = 1}^n {{c_{i,x}}\left( {x_0^{(4)}+{g_3}x_0^{\left( 3 \right)} + {g_2}{{\ddot x}_0} + {g_1}{{\dot x}_0}} \right)}  \\
       &~- \sum\limits_{i = 1}^n {\frac{{{g_{5,y}}c_{i,y}^2}}{{\left| {{c_{i,y}}} \right| + \gamma {e^{ - \lambda t}}}}}
       -\sum\limits_{i = 1}^n {{c_{i,y}}\left( {y_0^{(4)}+{g_3}y_0^{\left( 3 \right)} + {g_2}{{\ddot y}_0} + {g_1}{{\dot y}_0}} \right)}\\
       &\leq-g_4\|c\|^2 - \sum\limits_{i = 1}^n {\left( {\frac{{{g_{5,x}}c_{i,x}^2}}{{\left| {{c_{i,x}}} \right| + \gamma {e^{ - \lambda t}}}} - \left| {{c_{i,x}}} \right|{\sigma _{0,x}}} \right)} \\
        &~- \sum\limits_{i = 1}^n {\left( {\frac{{{g_{5,y}}c_{i,y}^2}}{{\left| {{c_{i,y}}} \right| + \gamma {e^{ - \lambda t}}}} - \left| {{c_{i,y}}} \right|{\sigma _{0,y}}} \right)},
\end{split}
\end{equation}
which, combined with the fact $W\geq \displaystyle\frac{\|c\|^2}{2\lambda_{\max}(\mathcal{H})}$, implies
\begin{equation}\label{dLsat}
\begin{split}
\dot{W}&\leq -{2g_4}{\lambda_{\min}(\mathcal{H})}W+\gamma n(\sigma_{0,x}+\sigma_{0,y})e^{-\lambda t}\\
       &= -qW+\sigma_0e^{-\lambda t},
\end{split}
\end{equation}
where $q:={2g_4}{\lambda_{\min}(\mathcal{H})},\sigma_0:=\gamma n(\sigma_{0,x}+\sigma_{0,y})$ and inequality $\displaystyle\frac{a_1x^2}{|x|+a_2}-a_3|x|\geq -a_2a_3$, with real numbers $a_1\geq a_3\geq 0$ and $a_2>0$, is applied. By comparison principle \cite{NLS}, integrating both sides of $(\ref{dLsat})$ results
\begin{equation}\label{Lsol}
W\left( t \right) \le \left\{ \begin{split}
{W_1}\left( t \right) &= {e^{ - qt}}W\left( 0 \right) + {\sigma _0}\frac{{{e^{ - qt}} - {e^{ - \lambda t}}}}{{\lambda  - q}},~\text{if}~q \neq \lambda; \\
{W_2}\left( t \right) &= {e^{ - qt}}W\left( 0 \right) + {\sigma _0}t{e^{ - qt}},~\text{if}~q = \lambda.
\end{split} \right.
\end{equation}
For $W_1(t)$, it satisfies
\begin{equation}\label{L1S}
{W_1}\left( t \right) \le {e^{ - qt}}W\left( 0 \right) + \frac{{2\sigma _0}}{{\left| {\lambda  - q} \right|}}{e^{ - \min \left\{ {q,\lambda } \right\}t}}.
\end{equation}
For $W_2(t)$, by the fact $te^{-a t} \leq \displaystyle\frac{1}{ae},\forall t\geq 0, a=0.5\lambda$, it satisfies
\begin{equation}\label{L2S}
{W_2}\left( t \right) \le {e^{ - qt}}W\left( 0 \right) + \frac{2\sigma _0}{{\lambda e}}{e^{ - \frac{\lambda }{2}t}}
\end{equation}
Hence, both $W_1$ and $W_2$ converge to zero globally exponentially, which means that $c$ and $s$ are globally exponentially stable (GES). By $(\ref{ccssx})$, one has $\xi^{(3)}=-g_3\ddot{\xi}-g_2\dot{\xi}-g_1\xi+s$. The vectors $\xi,\dot{\xi},\ddot{\xi}$ and $\xi^{(3)}$ are GES due to fact that both $\xi^{(3)}=-g_3\ddot{\xi}-g_2\dot{\xi}-g_1\xi$ and $s$-dynamics are GES\cite{RN19,RN80}. Therefore, the claims $(\ref{goa})$ follow.
\end{proof}

\end{document}